\documentclass[a4paper,USenglish]{lipics-v2016}

\usepackage{hyperref}

\usepackage{amsthm}
\usepackage[capitalize]{cleveref}
\usepackage{xparse}
\usepackage{xifthen}
\usepackage{todonotes}
\usepackage{xspace}
\usepackage{enumitem}

\theoremstyle{theorem}
\newtheorem{thm}{Theorem}[section]
\newtheorem{lem}[thm]{Lemma}
\newtheorem{prop}[thm]{Proposition}
\newtheorem{fact}[thm]{Fact}

\theoremstyle{remark}
\newtheorem{rem}[thm]{Remark}

\Crefname{lem}{Lemma}{Lemmas}
\Crefname{thm}{Theorem}{Theorems}
\Crefname{prop}{Proposition}{Propositions}
\Crefname{fact}{Fact}{Facts}
\Crefname{cor}{Corollary}{Corollaries}
\Crefname{rem}{Remark}{Remarks}

\newlist{axioms}{enumerate}{1}
\setlist[axioms]{label=(\roman*), widest=(iii), leftmargin=!, itemindent=!, labelindent=0pt}

\crefname{axiomsi}{axiom}{axioms}

\newlist{results}{enumerate}{1}
\setlist[results]{label=(\roman*), widest=(iii), leftmargin=!}
\crefname{resultsi}{result}{results}

\newcommand{\N}{\mathbb{N}}
\newcommand{\Ncal}{\mathcal{N}}
\newcommand{\Z}{\mathbb{Z}}
\newcommand{\Nom}{\overline{\mathbb{N}}}
\newcommand{\pP}[1][]{\mathfrak{p}_{\mathsf{#1}}}
\newcommand{\pPs}{\mathfrak{p}}

\newcommand{\cC}{\mathcal{C}}
\newcommand{\cA}{\mathcal{A}}
\newcommand{\cB}{\mathcal{B}}
\newcommand{\cS}{\mathcal{S}}

\newcommand{\dcl}[1]{#1\mathord{\downarrow}}

\newcommand{\Pre}{\textsc{Pre}}
\newcommand{\Post}{\textsc{Post}}
\newcommand{\trans}[1]{\xrightarrow{#1}}

\newcommand{\fin}{\textup{fin}}
\newcommand{\init}{\textup{init}}
\newcommand{\eff}{\Delta}

\newcommand{\eps}{\varepsilon}
\newcommand{\subword}{\preceq}

\newcommand{\NL}{\mathsf{NL}}
\newcommand{\PTIME}{\mathsf{PTIME}}
\newcommand{\EXPSPACE}{\mathsf{EXPSPACE}}

\title{Unboundedness problems for languages of vector addition systems}

\DeclareDocumentCommand{\factors}{O{} m}{\ifthenelse{\isempty{#1}}{F(#2)}{F_{#1}(#2)}}

\author[1]{Wojciech Czerwi\'{n}ski}
\author[1]{Piotr Hofman}
\author[2]{Georg Zetzsche\footnote{Supported by a fellowship of Fondation Sciences Math\'{e}matiques de Paris.}}
\affil[1]{University of Warsaw \\
\texttt{\{wczerwin,ph209519\}@mimuw.edu.pl}}
\affil[2]{IRIF (Universit\'{e} Paris-Diderot \& CNRS)\\
  \texttt{zetzsche@irif.fr}}

\authorrunning{W. Czerwi\'{n}ski and P. Hofman and G. Zetzsche} 

\Copyright{Wojciech Czerwi\'{n}ski and Piotr Hofman and Georg Zetzsche}

\EventEditors{Ioannis Chatzigiannakis, Christos Kaklamanis, Daniel Marx, and Don Sannella}
\EventNoEds{4}
\EventLongTitle{45th International Colloquium on Automata, Languages, and Programming (ICALP 2018)}
\EventShortTitle{ICALP 2018}
\EventAcronym{ICALP}
\EventYear{2018}
\EventDate{July 9--13, 2018}
\EventLocation{Prague, Czech Republic}
\EventLogo{eatcs}
\SeriesVolume{80}
\ArticleNo{}

\begin{document}
\maketitle
\begin{abstract}
  A vector addition system (VAS) with an initial and a final marking
  and transition labels induces a language. In part because the
  reachability problem in VAS remains far from being well-understood,
  it is difficult to devise decision procedures for such languages.
  This is especially true for checking properties that state
  the existence of infinitely many words of a particular
  shape. Informally, we call these \emph{unboundedness properties}.
  
  We present a simple set of axioms for predicates that can express
  unboundedness properties.  Our main result is that such a
  predicate is decidable for VAS languages as soon as it is decidable
  for regular languages. Among other results, this allows us to show
  decidability of (i)~separability by bounded regular languages,
  (ii)~unboundedness of occurring factors from a language $K$ with
  mild conditions on $K$, and (iii)~universality of the set of factors.
\end{abstract}


\section{Introduction}
Vector addition systems (VAS) and, essentially equivalent, Petri nets
are among the most widely used models of concurrent systems. Although
they are used extensively in practice, there are still fundamental
questions that are far from being well understood. 

This is reflected in what we know about decidability questions
regarding the most expressive class of languages associated to VAS:
The languages of (arbitrarily) labeled VAS with a given initial and
final configuration, which we just call \emph{VAS languages}.  In the
1970s, this class has been characterized in terms of closure
properties and Dyck languages by Greibach~\cite{Greibach1978} and
Jantzen~\cite{Jantzen1979}. Almost all decidability results about
these languages use a combination of these closure properties and the
decidability of the reachability problem for
VAS~\cite{DBLP:conf/stoc/Mayr81} (or for Reinhardt's
extension~\cite{Reinhardt2008}, such as
in~\cite{AtigGanty2011,Zetzsche2018a}). Of course, this method is
confined to procedures that somehow reduce to the existence of one or
finitely many runs of vector addition systems.

There are two notable exceptions (and, to the authors' knowledge,
these are the only exceptions) to this and they both rely on an
inspection of decision procedures for VAS.  The first is Hauschildt
and Jantzen's result~\cite{HauschildtJantzen1994} from 1994 that
finiteness of VAS languages is decidable, which employs Hauschildt's
algorithm to decide semilinearity of reachability
sets~\cite{hauschildt1990semilinearity}. The second is the much more
recent result of Habermehl, Meyer, and Wimmel from
2010~\cite{HabermehlMeyerWimmel2010}, showing that downward closures
are computable for VAS languages, which significantly generalizes
decidability of finiteness. Their proof involves a careful inspection
of marked graph-transition sequences (MGTS) of Lambert's algorithm for
the reachability proof.  This sparsity of decidability results is due
to the fact that the algorithms for the reachability problem are still
quite unwieldy and have been digested by few members of the research
community.

In particular, it currently seems difficult to decide whether there
exist infinitely many words of some shape in a given language---unless
the problem reduces to computing downward closures. Informally, we
call problems of this type \emph{unboundedness problems}. Such
problems are important for two reasons. The first concerns
\emph{separability problems}, which have attracted attention in recent
years~\cite{DBLP:journals/fuin/Bojanczyk17,DBLP:conf/stacs/ClementeCLP17,
DBLP:conf/icalp/Goubault-Larrecq16,DBLP:conf/csr/PlaceZ17,DBLP:conf/lics/PlaceZ17}.
Here, instead of deciding whether two languages are
disjoint, we are looking for a (typically finite-state) certificate
for disjointness, namely a set that includes one language and is
disjoint from the other.  For general topological reasons,
inseparability is usually witnessed by a common pattern, whose
presence in a language is an unboundedness property.  The second
reason is that unboundedness problems tend to be \emph{decidable where
  exact queries are not}. This phenomenon also occurs in the theory of
regular cost functions~\cite{Colcombet2013}.  Moreover, as it turns
out in this work, this is true for VAS languages as well.

\subparagraph*{Contribution} 
We present a simple notion of an \emph{unboundedness predicate} on
languages and show that such predicates are decidable for VAS
languages as soon as they are decidable for regular languages.  On the
one hand, this provides an easy and general way to obtain new
decidability results for VAS languages without the need to understand
the details of the KLMST decomposition. On the other hand, we apply
this framework to prove:
\begin{results}
\item\label{list:boundedness} Boundedness in the sense of Ginsburg and
  Spanier~\cite{ginsburg1964bounded} is decidable for VAS languages. A
  language $L\subseteq\Sigma^*$ is \emph{bounded} if there are
  $w_1,\ldots,w_n\in\Sigma^*$ with $L\subseteq w_1^*\cdots
  w_n^*$. Moreover, it is decidable whether two given VAS languages
  are separable by a bounded regular language.
\item\label{list:downward} Computability of downward closures
  can be recovered as well.
\item\label{list:factors} Suppose that $K\subseteq\Sigma^*$ is chosen
  so that it is decidable whether $K$ intersects a given regular
  language. Then, it is decidable for a given VAS language $L$ whether
  $L$ contains words with arbitrarily many factors from $K$.
  Moreover, in case the number of factor occurrences in $L$ is
  bounded, we can even compute an upper bound.
\item\label{list:universality} Under the same assumptions as above on
  $K\subseteq\Sigma^*$, one can decide if every word from $K^*$
  appears as a factor of a given VAS language $L\subseteq\Sigma^*$. In
  particular, it is decidable whether $L$ contains every word from
  $\Sigma^*$ as a factor.
\end{results}
It should be stressed that results
\labelcref{list:factors,list:universality} came deeply unexpected to
the authors. First, this is because the assumptions are already
satisfied when $K$ is induced by a system model as powerful as
well-structured transition systems or higher-order recursion schemes.
In these cases, it is in general \emph{undecidable} whether a given
VAS language contains a factor from $K$ \emph{at least once}, because
intersection emptiness easily reduces to this problem (see the remarks
after~\cref{non-overlapping-factors:decidable}). We therefore believe
that these results might lead to new approaches to verifying systems
with concurrency and (higher-order) recursion, where the latter
undecidability (or the unknown status in the case of simple
recursion~\cite{LerouxSutreTotzke2015}) is usually a barrier for
decision procedures.

The second reason for our surprise about
\labelcref{list:factors,list:universality} is that these problems are
undecidable as soon as $L$ is just slightly beyond the realm of VAS:
Already for one-counter languages $L$,
both \labelcref{list:factors,list:universality} become undecidable. Thus, compared to other infinite-state systems, VAS languages turn
out to be extraordinarily amenable to unboundedness problems.

\subparagraph*{Related work} Other authors have investigated general
notions of unboundedness properties for
VAS~\cite{DBLP:journals/ijfcs/AtigH11,BlockeletSchmitz2011,demri-infinity2010,yen1992unified},
usually with the goal of obtaining $\EXPSPACE$ upper bounds. However,
those properties a priori concern the state space itself. While they
can sometimes be used to reason about
languages~\cite{BlockeletSchmitz2011,demri-infinity2010}, this has
been confined to coverability languages, which are significantly less
expressive than the reachability languages studied here. Specifically,
every problem we consider here is hard for the reachability problem
(see \cref{complexity}).

An early attempt was Yen's work~\cite{yen1992unified}, which claimed
an $\EXPSPACE$ upper bound for a powerful logic concerning paths in
VAS.  Unfortunately, a serious flaw in the latter was discovered by
Atig and Habermehl~\cite{DBLP:journals/ijfcs/AtigH11}, who presented a
corrected proof for a restricted version of Yen's
logic. Demri~\cite{demri-infinity2010} then introduced a notion of
\emph{generalized unboundedness properties}, which covers more
properties from Yen's logic and proved an $\EXPSPACE$ procedure to
check them. Examples include reversal-boundedness, place boundedness,
and regularity of firing sequences of unlabeled VAS. Finally, Blockelet and
Schmitz~\cite{BlockeletSchmitz2011} introduce an extension of
computation tree logic (CTL) that can express ``coverability-like
properties'' of VAS. The authors prove an
$\EXPSPACE$ upper bound for model checking this logic on VAS.

\subparagraph{Organization} After \cref{sec:prelim} contains
preliminaries, \cref{sec:results} defines our notion of
unboundedness predicates and presents our main result.  In
\cref{sec:applications}, we apply the results to obtain the
consequences mentioned above. \Cref{sec:proof} is devoted to the
main result's proof.



\section{Preliminaries}\label{sec:prelim}
Let $\Sigma$ be a finite alphabet. For $w\in\Sigma^*$, we denote its length by $|w|$.
The $i$-th letter of $w$, for $i \in [1,|w|]$ is denoted $w[i]$.
Moreover, we write $\Sigma_\eps = \Sigma \cup \{\eps\}$.
A \emph{($d$-dimensional) vector addition system} (VAS) $V$ consists of finite set of \emph{transitions} $T \subseteq \Z^d$, 
\emph{source} and \emph{target} vectors $s, t \in \N^d$ and a \emph{labeling} $h\colon T \to \Sigma_\eps$, whose extension to a morphism $T^*\to\Sigma^*$ is also denoted $h$.
Vectors $v\in\N^d$ are also called \emph{configurations}. A transition $t \in T$ can be \emph{fired} in a configuration
$v \in \N^d$ if $v+t \in \N^d$. Then, the result of firing $t$ is the configuration $v+t$ and we write $v \trans{h(t)} v'$ for $v' = v+t$.  
For $w\in\Sigma^*$, we write $v \trans{w} v'$ if there exist $v_1, \ldots, v_k\in\N^d$
such that
$
v = v_0 \trans{x_1} v_1 \trans{x_2} \ldots \trans{x_k} v_k \trans{x_{k+1}} v_{k+1} = v',
$
where $w = x_1 \cdots x_{k+1}$ for some $x_1,\ldots,x_{k+1}\in\Sigma_\eps$. 
The \emph{language} of $V$, denoted $L(V)$, is the set of all labels of runs from source to target, i.e. $L(V) = \{w\in\Sigma^* \mid s \trans{w} t\}$.
The languages of the form $L(V)$ for VAS $V$ are called \emph{VAS languages}.
A word $u = a_1 \cdots a_n$ with $a_i \in \Sigma$ is a \emph{subword}
of a word $v \in \Sigma^*$ if $v \in \Sigma^* a_1 \Sigma^* \cdots
\Sigma^* a_n \Sigma^*$, which is denoted $u \subword v$.  For a
language $L \subseteq \Sigma^*$ its \emph{downward closure} is the
language $\dcl{L} = \{u\in\Sigma^* \mid \exists v \in L \colon u
\subword v\}$. It is known that $\dcl{L}$ is regular for every
$L\subseteq\Sigma^*$~\cite{Higman1952,Haines1969}.  A \emph{language
  class} is a collection of languages, together with some way of
finitely describing these languages (such as by grammars, automata,
etc.).  If $\cC$ is a language class so that given a description of a
language $L$ from $\cC$, we can compute an automaton for $\dcl{L}$, we
say that \emph{downward closures are computable} for $\cC$.

A \emph{full trio} is a language class that is effectively closed
under rational transductions~\cite{Berstel1979}, which are relations
defined by nondeterministic two-tape automata.  Examples of full trios
are abundant among infinite-state models: If a nondeterministic
machine model involves a finite-state control, the resulting language
class is a full trio. Equivalently, a full trio is a class that is
effectively closed under morphisms, inverse morphisms, and regular
intersection~\cite{Berstel1979}. Examples include \emph{VAS
  langauges}~\cite{Jantzen1979}, \emph{coverability languages of
  WSTS}~\cite{GeeraertsRaskinVanBegin2007}, \emph{one-counter languages}
(which are accepted by one-counter automata \emph{with zero
tests})~\cite{HopcroftUllman1979}, and languages of \emph{higher-order
pushdown automata}~\cite{Maslov1976} and \emph{higher-order recursion
schemes}~\cite{DBLP:conf/lics/HagueMOS08}. The context-sensitive do not
constitute a full trio, as they are not closed unter erasing morphisms.



\section{Main result}\label{sec:results}

Here, we introduce our notion of
unboundedness predicates and present our main result.

For didactic purposes, we begin our exposition of unboundedness
predicates with a simplified (but already useful) version.  An
important aspect of the definition is that technically, an
unboundedness predicates is not a property of the language
$L\subseteq\Sigma^*$ we want to analyze, but of the set of its
factors. In other words, we have a unary predicate $\pP$ on languages
and we want to decide whether $\pP(\factors{L})$, where
$\factors{L}=\{w\in\Sigma^* \mid L\cap \Sigma^*w\Sigma^*\ne\emptyset
\}$ is the set of factors of $L$. 
For the definition, it is helpful to keep in mind the simplest example
of an unboundedness predicate, the \emph{infinity predicate} $\pP[inf]$,
where $\pP[inf](K)$ if and only if $K$ is infinite. Then, $\pP[inf](\factors{L})$ if
and only if $L$ is infinite.
A unary predicate $\pP$ on languages over
$\Sigma^*$ is called \emph{1-dimensional unboundedness predicate} if
for every $K,L\subseteq\Sigma^*$, we have:
\begin{axioms}[label=({\roman*}$^*$), widest=(iii$^*$)]
\item\label{axiom:1dim:upward} if $\pP(K)$ and $K\subseteq L$, then $\pP(L)$.
\item\label{axiom:1dim:union} if $\pP(K\cup L)$, then either $\pP(L)$ or $\pP(K)$.
\item\label{axiom:1dim:concatenation} if $\pP(\factors{KL})$, then either $\pP(\factors{K})$ or $\pP(\factors{L})$.
\end{axioms}
Part of our result will be that for such predicates, if we can decide
whether $\pP(F(R))$ for regular languages $R$, we can decide whether
$\pP(F(L))$ for VAS languages $L$. Before we come to that, we want to
generalize a bit.
There are predicates we want to decide that fail to satisfy
\cref{axiom:1dim:concatenation}, such as the one stating
$a^*b^*\subseteq\dcl{L}$ for $L\subseteq\Sigma^*$: It is satisfied for
$a^*b^*$, but neither for $a^*$ nor for $b^*$. (Deciding such
predicates is useful for computing downward
closures~\cite{DBLP:conf/icalp/Zetzsche15} and separability by
piecewise testable
languages~\cite{DBLP:journals/dmtcs/CzerwinskiMRZZ17}) To capture such
predicates, which intuitively ask for several quantities being
unbounded simultaneously, we present a more general set of axioms.  Here, the idea is to
formulate predicates over simultaneously occurring factors.  For a
language $L\subseteq\Sigma^*$ and $n\in\N$, let
\[ \factors[n]{L} = \{(w_1,\ldots,w_n)\in(\Sigma^*)^n \mid \Sigma^*w_1\Sigma^*\cdots w_n\Sigma^*\cap L\ne\emptyset \}. \]
We will speak of \emph{$n$-dimensional} predicates, i.e., predicates
$\pP$ on subsets of $(\Sigma^*)^n$, and we want to decide whether
$\pP(\factors[n]{L})$ for a given language $L$.  The following are
axioms referring to all subsets $S,T\subseteq(\Sigma^*)^n$, languages
$L_i\subseteq\Sigma^*$, and all $k\in\N$.  We call $\pP$ an \emph{($n$-dimensional)
  unboundedness predicate} if
\begin{axioms}
\item\label{axiom:upward} if $\pP(S)$ and $S\subseteq T$, then $\pP(T)$.
\item\label{axiom:union} if $\pP(S\cup T)$, then $\pP(S)$ or $\pP(T)$.
\item\label{axiom:concatenation} if $\pP(\factors[n]{L_1\cdots L_k})$, then $n=n_1+\cdots +n_k$ such that
  $\pP(\factors[n_1]{L_1}\times\cdots\times\factors[n_k]{L_k})$.
\end{axioms}
Intuitively, the last axiom says that if a concatenation satisfies the
predicate, then this is already witnessed by factors in at most $n$
participants of the concatenation.  Note that for $n=1$, the axioms
coincide with the simplified
\cref{axiom:1dim:upward,axiom:1dim:union,axiom:1dim:concatenation}
above.  An $n$-dimensional unboundedness predicate $\pP$ is
\emph{decidable for a language class $\cC$} if, given a language $L$
from $\cC$, it is decidable whether $\pP(\factors[n]{L})$. The
following is our main result.

\begin{thm}\label{thm:approximation-unboundedness}
  Given a VAS language $L\subseteq\Sigma^*$, one can compute a regular
   $R\subseteq\Sigma^*$ such that $L\subseteq R$ and for every $n$-dim. unboundedness
  predicate $\pP$, we have $\pP(\factors[n]{L})$ iff
  $\pP(\factors[n]{R})$.
\end{thm}
Note that this implies that decidability of $\pP$ for regular
languages implies decidability of $\pP$ for VAS languages for any
$n$-dim. unboundedness predicate $\pP$.  In addition, when our
unboundedness predicate expresses that a certain quantity is
unbounded, then in the bounded case,
\cref{thm:approximation-unboundedness} sometimes allows us to compute
an upper bound (see, e.g. \cref{non-overlapping-factors:decidable}).

\begin{rem}\label{complexity}
  Let us comment on the complexity of deciding whether
  $\pP(\factors[n]{L})$ for a VAS language $L$. Call $\pP$
  \emph{non-trivial} if there is at least one $K\subseteq\Sigma^*$
  that satisfies $\pP$ and least one $K'\subseteq\Sigma^*$ for which
  $\pP$ is not satisfied. Then, deciding whether $\pP(\factors[n]{L})$
  is at least as hard as the reachability problem. Indeed, in this
  case \cref{axiom:upward} implies that
  $\factors[n]{\Sigma^*}=\Sigma^*$ satisfies $\pP$, but
  $\factors[n]{\emptyset}=\emptyset$ does not. Given a VAS $V$ and two
  vectors $\mu_1$ and $\mu_2$, it is easy to construct a VAS $V'$ so
  that $L(V')=\Sigma^*$ if $V$ can reach $\mu_2$ from $\mu_1$ and
  $L(V')=\emptyset$ otherwise.
\end{rem}



\section{Applications}\label{sec:applications}


\newcommand{\application}[1]{\subparagraph*{#1}}

\application{Bounded languages}
 
Our first application concerns bounded languages.  A language $L
\subseteq \Sigma^*$ is \emph{bounded} if there exist words $w_1,
\ldots, w_n \in \Sigma^*$ such that $L \subseteq w_1^* \cdots w_n^*$.
This notion was introduced by Ginsburg
and Spanier~\cite{ginsburg1964bounded}.  Since a
bounded language as above can be characterized by the set of vectors
$(x_1,\ldots,x_n)\in\N^n$ for which $w_1^{x_1}\cdots w_n^{x_n}\in L$,
bounded languages are quite amenable to analysis. This has led to a
number of applications to concurrent recursive
programs~\cite{DBLP:conf/popl/EsparzaG11,DBLP:conf/lics/EsparzaGM12,DBLP:journals/toplas/EsparzaGP14,DBLP:journals/fmsd/GantyMM12,long2012language},
but also counter systems~\cite{demri2010model} and WSTS~\cite{DBLP:journals/tcs/ChambartFS16}.

Boundedness has been shown decidable for \emph{context-free languages}
by Ginsburg and Spanier~\cite{ginsburg1964bounded}
($\PTIME$-completeness by
Gawrychowski~et~al.~\cite{gawrychowski2010finding}) and hence also for
\emph{regular languages} ($\NL$-completeness also
in~\cite{gawrychowski2010finding}), for \emph{equal matrix languages}
by Siromoney~\cite{siromoney1969characterization}, and for trace
languages of \emph{complete deterministic well-structured transition
  systems} by
Chambart~et~al.~\cite{DBLP:journals/tcs/ChambartFS16}. The latter
implies that boundedness is decidable for coverability languages of
deterministic vector addition systems, in which case
$\EXPSPACE$-completeness was shown by
Chambart~et~al.~\cite{DBLP:journals/tcs/ChambartFS16} (the upper bound
had been established by Blockelet and
Schmitz~\cite{BlockeletSchmitz2011}).

We use \cref{thm:approximation-unboundedness} to show the following.
\begin{thm}\label{thm:boundedness}
Given a VAS, it is decidable whether its language is bounded.
\end{thm}

The rest of this section is devoted to the proof of \cref{thm:boundedness}.
Let $\pP[notb]$ be the $1$-dimensional predicate that holds for a
language $K \subseteq \Sigma^*$ if and only if $K$ it is not
bounded. We plan to apply~\cref{thm:approximation-unboundedness} to $\pP[not]$, but it
allows us to decide only whether $\pP[notb](\factors{L})$ for a given VAS language $L$.
Thus we need the following fact, which we prove in a moment.
\begin{fact}\label{fact:boundedness-factors}
A language  $L\subseteq\Sigma^*$ is bounded if and only if $F(L)$ is bounded.
\end{fact}

Now we need to show that $\pP[notb]$ is indeed an unboundedness
predicate, meaning that it satisfies
\cref{axiom:1dim:upward,axiom:1dim:union,axiom:1dim:concatenation}.
By definition of boundedness, $\pP[notb]$ clearly fulfills
\cref{axiom:1dim:upward}: The subset of any bounded language is
bounded itself. \Cref{axiom:1dim:union,axiom:1dim:concatenation} are
implied by \cref{fact:boundedness-factors} and the following.
\begin{fact}\label{fact:boundedness-closed}
If $K$ and $L$ are bounded then both $K \cup L$ and $KL$ are bounded as well.
\end{fact}

\newcommand{\fofword}[1]{\factors{#1}}
Let us prove \cref{fact:boundedness-closed,fact:boundedness-factors} and
begin with \cref{fact:boundedness-closed}.  If $K$ and $L$ are
bounded, then $K \subseteq w_1^* \cdots w_n^*$ and $L \subseteq
w_{n+1}^* \cdots w_m^*$ for some $w_1, \ldots, w_m \in \Sigma^*$.
Then we have $K \cup L, KL \subseteq w_1^* \cdots w_m^*$, which shows
\cref{fact:boundedness-closed}.  In order to show
\cref{fact:boundedness-factors}, observe first that for each
individual word $w\in\Sigma^*$, the language $\fofword{w}$ is
bounded because it is finite.  Thus, if $L\subseteq w_1^*\cdots
w_n^*$, then $\factors{L}$ is included in 
$\fofword{w_1}w_1^*\fofword{w_1}\cdots \fofword{w_n}w_n^* \fofword{w_n}$,
which is bounded as a concatenation of bounded languages by
\cref{fact:boundedness-closed}.  Thus, $F(L)$ is bounded as
well. Conversely, $L$ inherits boundedness from its superset
$\factors{L}$.

To conclude  \cref{thm:boundedness}, we need to show that given regular language
$R\subseteq\Sigma^*$, it is decidable whether
$\pP[notb](\factors{R})$.  By
\cref{fact:boundedness-factors}, this amounts to checking whether $R$
is bounded. This is decidable even for context-free
languages~\cite{ginsburg1964bounded} (and in $\NL$ for regular
ones~\cite{gawrychowski2010finding}).

\subparagraph*{Separability}
We can also use our results to decide
whether two VAS languages are separable by a bounded regular
language. Very generally, if $\cS$ is a class of sets, we say that a
set $K$ is \emph{separable from} a set $L$ \emph{by} a set from $\cS$
if there is a set $S$ in $\cS$ so that $K\subseteq S$ and $L\cap
S=\emptyset$. 

The separability problem was recently investigated for VAS languages
and several subclasses thereof.
In~\cite{DBLP:journals/dmtcs/CzerwinskiMRZZ17} it is shown that
separability of VAS languages by piecewise testable languages (a
subclass of regular languages) is decidable.  Decidability of separability of VAS
languages by regular languages is still open, but it is
known for several subclasses of VAS
languages~\cite{DBLP:conf/icalp/ClementeCLP17,DBLP:conf/stacs/ClementeCLP17,DBLP:conf/lics/CzerwinskiL17}.
In~\cite{DBLP:journals/corr/CzerwinskiL17a} it is shown that any two
disjoint VAS coverability languages are separable by a regular
language.  Here, using~\cref{thm:boundedness} we are able to show the following.

\begin{thm}\label{thm:bounded-separability}
Given two VAS languages $K$ and $L$, it is decidable whether $K$ is separable
from $L$ by a bounded regular language.
\end{thm}

Clearly, in order for that to hold, $K$ has to be bounded, which we
can decide. Moreover, by enumerating expressions $w_1^*\cdots w_n^*$,
we can find one with $K\subseteq w_1^*\cdots w_n^*$. Since the bounded regular languages (BRL) are
closed under intersection (recall that a subset of a bounded language
is again bounded), $K$ and $L$ are separable by a BRL if and only if
$L_0=K$ and $L_1=L\cap w_1^*\cdots w_n^*$ are separable by a BRL. Since now
both input languages are included in $w_1^*\cdots w_n^*$, we can reformulate
the problem into one over vector sets.
\begin{lem}\label{separability-commutative}
  Let $L_0,L_1\subseteq w_1^*\cdots w_n^*$ and
  $U_i=\{(x_1,\ldots,x_n)\in\N^n \mid w_1^{x_1}\cdots w_n^{x_n}\in
  L_i\}$ for $i\in\{0,1\}$.  Then, $L_0$ is separable from $L_1$ by a
  BRL if and only if $U_0$ is separable from $U_1$ by a recognizable
  subset of $\N^n$.
\end{lem}
Recall that a subset $S\subseteq\N^n$ is \emph{recognizable} if there is a morphism
$\varphi\colon\N^n\to F$ into a finite monoid $F$ with $S=\varphi^{-1}(\varphi(S))$.
\Cref{separability-commutative} is a straightforward application of Ginsberg and Spanier's
characterization of BRL~\cite{GinsburgSpanier1966a}.

Since in our case, $L_0$ and $L_1$ are VAS languages, a standard construction shows that $U_0$ and $U_1$
are (effectively computable) sections of VAS reachability sets. 
Here, sections are defined as follows. For a subset $I\subseteq [1,n]$,
let $\pi_I\colon \N^n\to\N^{|I|}$ be the projection onto the coordinates in $I$.
Then, every set of the form $\pi_{[1,n]\setminus I}(S\cap \pi_{I}^{-1}(x))$ for some
$I\subseteq[1,n]$ and $x\in\N^{|I|}$ is called a \emph{section} of $S\subseteq\N^n$.
Thus, the following result by Clemente~et~al.~\cite{DBLP:conf/stacs/ClementeCLP17} allows us to decide separability by BRL.
\begin{thm}[\cite{DBLP:conf/stacs/ClementeCLP17}]
  Given two sections $S_0,S_1\subseteq\N^n$ of reachability sets of
  VAS, it is decidable whether $S_0$ is separable from $S_1$ by a
  recognizable subset of $\N^n$.
\end{thm}


  

\application{Downward closures and simultaneus unboundedness}
We now illustrate how to compute downward closures using our results.
First of all, computability of downward closures for VAS languages
follows directly from \cref{thm:approximation-unboundedness} because
it implies $\dcl{R}=\dcl{L}$: For each word $w=a_1\cdots a_n$ with
$a_1,\ldots,a_n\in\Sigma$, consider the $n$-dimensional predicate
$\pPs_w$ which is satisfied for $S\subseteq(\Sigma^*)^n$ iff
$(a_1,\ldots,a_n)\in S$. Then $\pPs_w(\factors[n]{L})$ if and only if
$w\in\dcl{L}$. It is easy to check that this is an unboundedness
predicate. Hence, $\dcl{R}=\dcl{L}$.

However, in order to illustrate how to apply unboundedness predicates,
we present an alternative approach. In
\cite{DBLP:conf/icalp/Zetzsche15}, it was shown that if a language
class $\cC$ is closed under rational transductions (which is the case
for VAS languages), then downward closures are computable for $\cC$ if
and only if, given a language $L$ from $\cC$ and letters
$a_1,\ldots,a_n$, it is decidable whether $a_1^*\cdots
a_n^*\subseteq\dcl{L}$.  Let us show how to decide the latter using
unboundedness predicates.

For this, we use an $n$-dimensional predicate. For a subset
$S\subseteq(\Sigma^*)^n$, let $\dcl{S}$ be the set of all tuples
$(u_1,\ldots,u_n)\in(\Sigma^*)^n$ such that there is some
$(v_1,\ldots,v_n)\in S$ with $u_i\subword v_i$ for $i\in[1,n]$.  Our
predicate $\pP[sup]$ is satisfied for $S\subseteq(\Sigma^*)^n$ if and
only if $a_1^*\times\cdots \times a_n^*\subseteq S$. Then clearly
$\pP[sup](\factors[n]{L})$ if and only if $a_1^*\cdots
  a_n^*\subseteq\dcl{L}$.  It is easy to check that $\pP$ fulfills
  \cref{axiom:upward} and \cref{axiom:union}. For the latter, note
  that $a_1^*\times\cdots\times a_n^*\subseteq \dcl{(S_1\cup S_2)}$
  implies that for some $j\in\{1,2\}$, there are infinitely many
  $\ell\in\N$, with $(a_1^\ell,\ldots,a_n^\ell)\in S_j$ and hence
  $a_1^*\times\cdots\times a_n^*\subseteq\dcl{S_j}$. For \cref{axiom:concatenation},
we need a simple combinatorial argument:
\begin{lem}\label{lemma:sup}
If $a_1^*\times\cdots \times a_n^*\subseteq\dcl{\factors[n]{L_1\cdots L_k}}$, then $n=n_1+\cdots+n_k$
with $a_1^*\times\cdots\times a_n^*\subseteq\dcl{(\factors[n_1]{L_1}\times\cdots\times\factors[n_k]{L_k})}$.
\end{lem}

It remains to show that for a regular language $R$, it is decidable
whether $a_1^* \cdots a_n^* \subseteq \dcl{R}$. Since it is easy to
construct an automaton for $\dcl{R}$, this amounts to a simple
inclusion check.

\application{Non-overlapping factors}
Our next example shows that under very mild assumptions on a
language $K$, one can decide whether the words in a VAS language $L$
contain arbitrarily many factors from $K$.  For $w\in\Sigma^*$ and
$K\subseteq\Sigma^+$, let $|w|_K$ be the largest number $m$ such that
there are $w_1,\ldots,w_m\in K$ with $(w_1,\ldots,w_m)\in
\factors[m]{w}$. Note that since $\eps \not\in K$, there is always
a maximal such $m$.  Consider the function $f_K\colon\Sigma^*\to\N$,
$w\mapsto |w|_K$. A function $f\colon \Sigma^*\to\Nom$ is
\emph{unbounded} on $L\subseteq\Sigma^*$ if for every $k\in\N$,
we have $f(w)\ge k$ for some $w\in L$.

\begin{thm}\label{non-overlapping-factors:decidable}
  If $\cC$ is a full trio with decidable emptiness problem, then given
  a VAS language $L$ and a language $K\subseteq\Sigma^+$ from $\cC$,
  it is decidable whether $f_K$ is unbounded on $L$. If $f_K$ is
  bounded on $L$, we can compute an upper bound.
\end{thm}

\Cref{non-overlapping-factors:decidable} is quite
unexpected because very slight variations lead to undecidability.
%
If we ask whether $f_K$ is non-zero on a given VAS language (as
opposed to unbounded), then this is in general undecidable. Indeed,
suppose $\cC$ is a full trio for which intersection with VAS languages
is undecidable (such as languages of lossy channel systems\footnote{It
  seems to be folklore that intersection between languages of lossy
  channel systems and languages of one-dimensional VAS is undecidable
  (the additional counter can be used to ensure that no letter is
  dropped). The only reference we could find is \cite{Reinhardt2015}.}
or higher-order pushdown
languages~\cite{HagueLin2011,DBLP:conf/icalp/Zetzsche15}). Then given
a language $K\subseteq\Sigma^*$ from $\cC$, a VAS language $L$ and
some $c\notin\Sigma$, the function $f_{cKc}$ is non-zero on $cLc$ if
and only if $K\cap L\ne\emptyset$.

Furthermore, the same problem becomes undecidable in general if
instead of VAS languages, we want to decide the problem for a language
class as simple as one-counter languages (OCL). Indeed, suppose $\cC$
is a full trio for which intersection with OCL is undecidable (such as
the class of OCL). For a given $K\subseteq\Sigma^*$ from $\cC$, an OCL
$L\subseteq\Sigma^*$, and some $c\notin\Sigma$, the set $c(Lc)^*$ is
effectively an OCL and $f_{cKc}$ is unbounded on $c(Lc)^*$ if and only
if $K\cap L\ne\emptyset$.

Let us prove \cref{non-overlapping-factors:decidable}. Fix a language
$K\subseteq\Sigma^*$ from $\cC$.  Our predicate $\pP[nof]$ is
one-dimensional and is satisfied on a set $L\subseteq\Sigma^*$ if and
only if $f_K$ is unbounded on $L$. Then clearly,
$\pP[nof](\factors{L})$ if and only if $f_K$ is unbounded on $L$.  It
is immediate that \cref{axiom:1dim:upward,axiom:1dim:union} are
satisfied. Furthermore, \cref{axiom:1dim:concatenation} follows by
contraposition: If neither $\pP[nof](\factors{L_0})$ nor
$\pP[nof](\factors{L_1})$, then there are $B_0,B_1\in\N$ such that
$f_K$ is bounded by $B_i$ on $L_i$ for $i=0,1$. That implies that
$f_K$ is bounded by $B_0+B_1+1$ on $L_0L_1$. This rules out
$\pP[nof](\factors{L_0L_1})$, which establishes
\cref{axiom:1dim:concatenation}. The following uses standard arguments. 
\begin{lem}\label{factors-unboundedness-regular}
  Let $\cC$ be a full trio with decidable emptiness problem. Given a
  language $K$ from $\cC$ and a regular language $R$, it is decidable
  whether $f_K$ is unbounded on $R$. Moreover, if $f_K$ is bounded on
  $R$, we can compute an upper bound.
\end{lem}
We can deduce \cref{non-overlapping-factors:decidable} from
\cref{factors-unboundedness-regular} as follows. Using
\cref{thm:approximation-unboundedness}, we compute the language $R$.
Then, $f_K$ is unbounded on $R$ iff it is unbounded on $L$. Moreover,
an upper bound for $f_K$ on $R$ is also an upper bound for $f_K$ on
$L$ because $L\subseteq R$.

\application{Counting automata} To illustrate how these results can be
used, we formulate an extension of
\cref{non-overlapping-factors:decidable} in terms of automata that can
count. Let $\cC$ be a full trio. Intuitively, a $\cC$-counting automaton can read a word
produced by a VAS and can use machines corresponding to $\cC$ as oracles.
Just like the intersection of two languages that describe threads in a
concurrent system signals a safety
violation~\cite{BouajjaniEsparzaTouili2003,ChakiEtAl2006,long2012language},
a successful oracle call would signal a particular undesirable event.
In such a model, it would be undecidable whether any oracle
call can be successful if, for example, $\cC$ is the class of
higher-order pushdown languages. However, we show that it
\emph{is} decidable whether such an automaton can make an unbounded
number of successful oracle calls and if not, compute an upper bound.
Hence, we can decide if the number of undesirable events is bounded
and, if so, provide a bound.

\newcommand{\oppush}[1]{\mathsf{push}(#1)}
\newcommand{\opcheck}[2]{\mathsf{check}(#1,#2)}
\newcommand{\Op}{\Omega}


A \emph{$\cC$-counting automaton} is a tuple
$\cA=(Q,\Sigma,\Gamma,C,q_0,E,Q_f)$, where $Q$ is a finite set of
\emph{states}, $\Sigma$ is its \emph{input alphabet}, $\Gamma$ is its
\emph{(oracle) tape alphabet}, $C$ is a finite set of \emph{counters},
$q_0\subseteq Q$ is its \emph{initial state}, $Q_f\subseteq Q$ is its
set of \emph{final states}, and $E\subseteq Q\times\Sigma^*\times
(\Op\cup\{\varepsilon\})\times Q$ is a finite set of \emph{edges},
where $\Op$ is a set of \emph{operations} of the following
form. First, we have an operation $\oppush{a}$ for each
$a\in\Gamma$, which appends $a$ to the oracle tape. Moreover,
we have $\opcheck{K}{c}$ for each $K\subseteq\Gamma^*$ from $\cC$ and
each $c\in C$, which first checks whether the current tape content
belongs to $K$ and if so, increments the counter $c$. After the oracle query,
it empties the oracle tape, regardless of whether the oracle anwsers
positively or negatively.

\newcommand{\autstep}[1]{\xrightarrow{#1}}
\newcommand{\reset}[2]{#1[#2\mapsto 0]}
A \emph{configuration} of $\cA$ is a triple $(q,u,\mu)$, where $q\in
Q$ is the current state, $u\in\Gamma^*$ is the oracle tape content, and $\mu\in\N^C$ describes
the counter values.  For a label
$x\in\Sigma\cup\{\varepsilon\}$, and configurations $(q,u,\mu),
(q',u',\mu')$, we write $(q,u,\mu)\autstep{x}(q',u',\mu')$ if $(q',u',\mu')$ results from
$(q,u,\mu)$ as described above.
%
%
%
In the general case $w\in\Sigma^*$,  $(q,u,\mu)\autstep{w}(q',u',\mu')$ has the obvious meaning.
%
$\cA$ defines a function $\Sigma^*\to\Nom$:
\[ \cA(w)=\sup\left\{\left.\inf_{c\in C} \mu(c) ~\right|~ \mu\in\N^C,~(q_0,\varepsilon,0)\xrightarrow{w} (q,u,\mu) ~\text{for some $q\in Q_f,~u\in\Gamma^*$}\right\}. \]
%
Hence, $\cA$ is unbounded on $L$ if for every $k\in\N$, there is a
$w\in L$ and a run of $\cA$ on $w$ in which for each $c\in C$, at
least $k$ of the oracle queries for $c$ are successful. The following
can be shown similarly to \cref{non-overlapping-factors:decidable},
but using a multi-dimensional unboundedness predicate.
\begin{thm}\label{counting-automata:decidable}
  Let $\cC$ be a full trio with decidable emptiness.  Given a VAS
  language $L$ and a $\cC$-counting automaton $\cA$, it is decidable
  whether $\cA$ is unbounded on $L$. Moreover, if $\cA$ is bounded on
  $L$, then one can compute an upper bound $B\in\N$ for $\cA$ on $L$.
\end{thm}

\application{Factor inclusion} As a last example, we show how our
results can be used to decide inclusion problems. Specifically, given
a VAS language $L\subseteq\Sigma^*$, it is decidable whether
$\Sigma^*\subseteq\factors{L}$. In fact, we show a more general result:





\begin{thm}\label{factor-universality:ext:decidable}
  If $\cC$ is a full trio with decidable emptiness problem, then given
  a VAS language $L$ and a language $K$ from $\cC$, it is decidable
  whether $K^*\subseteq\factors{L}$.
\end{thm}
Here, $\Sigma^*\subseteq\factors{L}$ is the special case where $K=\Sigma$.
Recall that is is undecidable whether
$L=\Sigma^*$ for VAS languages and for
one-counter languages (OCL) (e.g.~\cite[Lemma
6.1]{DBLP:journals/dmtcs/CzerwinskiMRZZ17}).

Similar to \cref{non-overlapping-factors:decidable}, deciding whether
$\Sigma^*\subseteq\factors{L}$ is already undecidable for OCL
$L$: For a given OCL $L\subseteq\Sigma^*$, pick a letter
$c\notin\Sigma$ and note that
$L'=c(Lc)^*\subseteq(\Sigma\cup\{c\})^*$ is effectively an OCL and
$(\Sigma\cup\{c\})^*\subseteq\factors{L'}$ if and only if
$L=\Sigma^*$.
Also, under the assumptions of the
\lcnamecref{factor-universality:ext:decidable}, it is undecidable
whether $K\subseteq\factors{L}$: If $L\subseteq\Sigma^*$ and
$c\notin\Sigma$, then $c\Sigma^*c\subseteq\factors{cLc}$ if and only
if $L=\Sigma^*$ (every full trio contains the regular set
$c\Sigma^*c$).

Let us see how \cref{factor-universality:ext:decidable} follows from
\cref{thm:approximation-unboundedness}.  Fix a language $K$ from $\cC$.  We use the
$1$-dim. predicate $\pP[fu]$, which is satisfied on a set
$L\subseteq\Sigma^*$ if and only if $K^*\subseteq\factors{L}$. Of
course, \cref{axiom:upward} holds by definition.
\Cref{axiom:concatenation} follows by contraposition: Suppose that
$K^*\subseteq\factors{L_1L_2}$ and $K^*\not\subseteq\factors{L_1}$
with some $u\in K^*\setminus \factors{L_1}$.  Let $v\in K^*$ be
arbitrary. Then, since $K^*\subseteq\factors{L_1L_2}$, we have
$uv\in\factors{L_1L_2}$. This means, there are $x,y\in\Sigma^*$ with
$xuvy\in L_1L_2$.  Hence, we have $xuvy=w_1w_2$ for some $w_i\in L_i$
for $i=1,2$. Then $|w_1|<|xu|$, because otherwise $u$ would belong to
$\factors{L_1}$. Therefore, $v$ is a factor of $w_2$ and thus
$v\in\factors{L_2}$. Hence, $K^*\subseteq\factors{L_2}$. Of course, a
similar argument works if $K^*\subseteq\factors{L_1L_2}$ and
$K^*\not\subseteq\factors{L_2}$. This proves
\cref{axiom:concatenation}. \Cref{axiom:union} can be shown the same
way. Thus, by \cref{thm:approximation-unboundedness}, it
  suffices to decide whether $K^*\subseteq\factors{R}$ for regular
  $R$, which follows from $\cC$ being a full trio and having decidable
  emptiness (see \cref{factor-universality-regular}).



\section{Proof of the main result}\label{sec:proof}
We prove our decidability result using the KLMST decomposition. More
specifically, we show a consequence that might be interesting in its
own right.

\begin{thm}\label{thm:approximation}
Given a VAS language $L\subseteq\Sigma^*$, one can compute $m, k\in\N$ and regular languages 
$R_{i,j}\subseteq\Sigma^*$, for $i \in [1,m]$, $j\in [1,k]$ so that
\begin{align}
L\subseteq\bigcup_{i=1}^m R_{i,1}\cdots R_{i,k} && \text{and} && R_{i,1}\times \cdots \times R_{i,k}\subseteq \factors[k]{L}~\text{for every $i\in[1,m]$}.\label{approx-rel}
\end{align}
\end{thm}

We first show how to derive \cref{thm:approximation-unboundedness}
from \cref{thm:approximation} and then proceed with the proof of
\cref{thm:approximation}, as it is much more technically complicated.

\subparagraph*{Proof of \cref{thm:approximation-unboundedness}}
Suppose \cref{thm:approximation} holds. Then, given a VAS language
$L$, we compute $m, k\in\N$ and the regular languages $R_{i,j}$ for
$i\in[1, m], j\in[1, k]$. We choose $R=\bigcup_{i=1}^m R_{i,1}\cdots
R_{i,k}$.  Then we have $L\subseteq R$. Let us show that
$\pP(\factors[n]{L})$ if and only if $\pP(\factors[n]{R})$. If
$\pP(\factors[n]{L})$, then clearly $\pP(\factors[n]{R})$, because
$L\subseteq R$ implies $\factors[n]{L}\subseteq \factors[n]{R}$ and by
\cref{axiom:upward}, this implies $\pP(\factors[n]{R})$. Conversely,
suppose $\pP(\factors[n]{R})$. Then by \cref{axiom:union}, there is an
$i\in[1,m]$ such that $\pP(\factors[n]{R_i})$, where
$R_i=R_{i,1}\cdots R_{i,k}$. According to \cref{axiom:concatenation},
we can write $n=n_1+\cdots+n_k$ such that $\pP$ holds for
$S:=\factors[n_1]{R_{i,1}}\times \cdots\times\factors[n_k]{R_{i,k}}$.
Note that by the choice of $R_{i,j}$, we have $R_{i,1}\times\cdots\times R_{i,k}\subseteq
\factors[k]{L}$ and therefore $S\subseteq
\factors[n]{L}$. This implies $\pP(\factors[n]{L})$ by \cref{axiom:upward}.

\subparagraph*{Proof of \cref{thm:approximation}}

The remainder of this section is devoted to the proof of \cref{thm:approximation}.
Like the method for computing downward closures by Habermehl, Meyer, and
Wimmel~\cite{HabermehlMeyerWimmel2010}, the construction of the sets $R_{i,j}$ is
based on Lambert's proof~\cite{DBLP:journals/tcs/Lambert92} of the decidability of
the reachability problem for Petri nets. In order to be compatible with Lambert's
exposition, we phrase our proof in terms of Petri nets instead of vector addition
systems.

A \emph{Petri net} $N = (P, T, \Pre, \Post)$ consists of a finite set $P$ of \emph{places},
a finite set $T$ of \emph{transitions} and two mappings $\Pre, \Post\colon T \to \N^P$.
Configurations of Petri net are elements of $\N^P$, called \emph{markings}.
For two markings $M, M'$ we say that $M'$ \emph{dominates} $M$, denoted $M\leq M'$, if for every place $p \in P$,
we have $M[p] \leq M'[p]$. The \emph{effect} of a transition $t \in T$
is $\Post(t) - \Pre(t) \in \Z^P$, denoted $\eff(t)$. If a marking $M$ dominates $\Pre(t)$ for a transition $t \in T$ then
$t$ is \emph{fireable} in $M$ and the result of firing $t$ in marking $M$ is $M' = M + \eff(t)$,
we write $M \trans{t} M'$. We extend notions of fireability and firing naturally to sequences
of transitions, we also write $M \trans{w} M'$ for $w \in T^*$. The \emph{effect} of $w \in T^*$
is sum of the effects of its letters, $\eff(w) = \sum_{i = 1}^{|w|} \eff(w[i])$.

For a Petri net $N = (P, T, \Pre, \Post)$ and markings $M_0,M_1$, we define the language 
$L(N, M_0, M_1) = \{w \in T^* \mid M_0 \trans{w} M_1\}$.
Hence, $L(N,M_0,M_1)$ is the set of transition sequences leading from
$M_0$ to $M_1$.  Moreover, let $L(N, M_0) = \bigcup_{M \in \N^P} L(N,
M_0, M)$, i.e. the set of all the transition sequences fireable in $M_0$.  A
\emph{labeled Petri net} is a Petri net $N=(P,T,\Pre,\Post)$ together
with an \emph{initial marking} $M_I$, a \emph{final marking} $M_F$,
and a \emph{labeling}, i.e.  a homomorphism $T^*\to\Sigma^*$.
The language \emph{recognized by} the labeled Petri net is then defined as
$L_h(N, M_I, M_F) = h(L(N, M_I, M_F))$.

It is folklore (and easy to see) that a language is a VAS language if and only
if it is recognized by a labeled Petri net (and the translation is
effective). Thus, it suffices to show \cref{thm:approximation} for languages of
the form $L=h(L(N,M_I,M_F))$.  Moreover, it is already enough to prove
\cref{thm:approximation} for languages of the form $L(N,M_I,M_F)$. Indeed,
observe that if we have constructed $R_{i,j}$ so that \cref{approx-rel} is
satisfied, then with $S_{i,j}=h(R_{i,j})$, we have
$h(L)\subseteq\bigcup_{i=1}^m S_{i,1}\cdots S_{i,k}$
and
$S_{i,1}\times\cdots \times S_{i,k}\subseteq \factors[k]{h(L)}$ for every $i\in[1,m]$.
Thus from now on, we assume $L=L(N,M_I,M_F)$ for a fixed Petri net $N=(P,T,\Pre,\Post)$.

%

\subparagraph*{The KLMST decomposition}
Lambert's decision procedure~\cite{DBLP:journals/tcs/Lambert92} is a refinement of
the previous ones by Mayr~\cite{DBLP:conf/stoc/Mayr81} and
Kosaraju~\cite{DBLP:conf/stoc/Kosaraju82}. Later, Leroux and
Schmitz~\cite{leroux2015reachability}  recast it again as an algorithm using WQO
ideals and dubbed the procedure \emph{KLMST
  decomposition} after its
inventors~\cite{DBLP:conf/stoc/Kosaraju82,DBLP:journals/tcs/Lambert92,DBLP:conf/stoc/Mayr81,sacerdote1977decidability}.

The idea is the following. We disregard for a moment that a transition sequence
has to keep all intermediate markings non-negative and only look for a sequence
that may go negative on the way.  It is standard technique to express the
existence of such a sequence as a linear equation system $Ax=b$. As expected,
solvability of this system is not sufficient for the existence of an actual
run. However, if we are in the situation that we can find (a)~runs that pump up
all coordinates arbitrarily high and also (b)~counterpart runs that remove those
excess tokens again, then solvability of the equation system is also sufficient:
We first increase all coordinates high enough, then we execute our
positivity-ignoring sequence, and then we pump down again. Roughly speaking, the
achievement of the KLMST decomposition is to put us in the latter situation, which
we informally call \emph{perfect circumstances}.

To this end, one uses a data structure, in Lambert's version called \emph{marked
  graph-transition sequence (MGTS)}, which restricts the possible runs of the
Petri net.  If the MGTS satisfies a condition that realizes the above perfect
circumstances, then it is called \emph{perfect}. Unsurprisingly, not every MGTS is
perfect.  However, part of the procedure is a decomposition of an imperfect MGTS
into finitely many MGTS that are less imperfect. Moreover, this
decomposition terminates in a finite set of perfect MGTS.  Thus, applied to an
MGTS whose restriction is merely to start in $M_I$ and end in $M_F$, then the
decomposition yields finitely many perfect MGTS $\Ncal_1,\ldots,\Ncal_n$ such that
the runs from $M_I$ to $M_F$ are precisely those conforming to at least one
of the MGTS. Moreover, checking whether $\Ncal_i$ admits a run amounts to
solving a linear equation system.


\subparagraph*{Basic notions} Let us introduce some notions used in Lambert's proof.
We extend the set of configurations $\N^d$ into $\Nom^d$, where $\Nom = \N \cup \{\omega\}$ for $\omega$ being
the first infinite ordinal number and representing the infinity. We extend the notion of transition firing into $\Nom^d$ naturally, by
defining $\omega - k = \omega = \omega + k$ for every $k \in \N$. For $u, v \in \Nom^d$ we write $u \leq_\omega v$ if
$u[i] = v[i]$ or $v[i] = \omega$. Intuitively reaching a configuration with $\omega$ at some places
means that it is possible to reach configurations with values $\omega$ substituted by arbitrarily high values.

A key notion in~\cite{DBLP:journals/tcs/Lambert92} is that of MGTS, which formulate restrictions on paths in
Petri nets. A \emph{marked graph-transition sequence (MGTS)} for our Petri net
$N=(P,T,\Pre,\Post)$ is a finite sequence
$
C_0, t_1, C_1 \ldots C_{n-1}, t_n, C_n,
$
where $t_i$ are transitions from $T$ and $C_i$ are precovering graphs, which are defined next.
A \emph{precovering graph} is a quadruple $C = (G, m, m^\init, m^\fin)$, where $G=(V,E,h)$ is a finite, strongly connected, directed graph with $V\subseteq\Nom^P$ and labeling $h\colon E \to T$,
and three vectors: a \emph{distinguished} vector $m \in V$, an \emph{initial} vector $m^\init \in \Nom^P$, and a \emph{final} vector $m^\fin \in \Nom^P$.
A precovering graph has to meet two conditions:
First, for every edge $e = (m_1, m_2) \in E$, there is an $m_3\in\Nom^P$ with $m_1 \trans{h(e)} m_3 \leq_\omega m_2$.
Second, we have
$m^\init, m^\fin \leq_\omega m$.
Additionally we impose the restriction on MGTS that the initial vector of $C_0$
equals $M_I$ and the final vector of $C_n$ equals $M_F$.

\subparagraph*{Languages of MGTS}
Each precovering graph can be treated as a finite automaton. For $m_1,m_2\in V$,
we denote by $L(C,m_1,m_2)$ the set of all  $w\in T^*$ read on a path from
$m_1$ to $m_2$.  Moreover, let $L(C) = L(C, m, m)$.  MGTS have associated
languages as well.  Let $\Ncal = C_0, t_1, C_1 \ldots C_{n-1}, t_n, C_n$ be an
MGTS of a Petri net $N$, where $C_i = (G_i, m_i, m_i^\init, m_i^\fin)$. Its
language $L(\Ncal)$ is the set of all words of the form $w = w_0 t_1 w_1 \cdots
w_{n-1} t_n w_n \in T^*$ where:
$w_i \in L(C_i)$ for each $i\in[0,n]$ 
and 
(ii)~there exist markings $\mu_0, \mu'_0, \mu_1, \mu'_1, \ldots, \mu_n, \mu'_n \in \N^P$ such
  that $\mu_i \leq_\omega m^\init_i$ and $\mu'_i \leq_\omega m^\fin_i$ and
  $\mu_0 \trans{w_0} \mu'_0 \trans{t_1} \mu_1 \trans{w_1} \ldots \trans{w_{n-1}} \mu'_{n-1} \trans{t_n} \mu_n \trans{w_n} \mu'_n$.
Notice that by (ii) and the restriction that $m^\init_0 = M_I$ and $m^\fin_n = M_F$,
we have $L(\Ncal)\subseteq L(N,M_I,M_F)$ for any MGTS $\Ncal$.

Hence roughly speaking, $L(\Ncal)$ is the set of runs that contain the transitions
$t_1,\ldots,t_n$ and additionally markings before and after firing these
transitions are prescribed on some places: this is exactly what the restrictions
$\mu_i \leq_\omega m^\init_i$, $\mu'_i \leq_\omega m^\fin_i$ impose.

Notice that at the moment we do not expect that values $\omega$ occurring at $m_i,
m_i^\init, m_i^\fin$ impose any restriction on the form of accepted runs.  Meaning
of $\omega$ values is reflected in the notion of \emph{perfect} MGTS described
later.
As an immediate consequence of the definition, we observe that for every
MGTS $\Ncal = C_0, t_1, C_1 \ldots C_{n-1}, t_n, C_n$
we have 
\begin{equation} L(\Ncal) \subseteq L(C_0) \cdot \{t_1\} \cdot L(C_1) \cdots L(C_{n-1}) \cdot \{t_n\} \cdot L(C_n). \label{eq:overapproximation}\end{equation}

\subparagraph{Perfect MGTS} As announced above, Lambert calls MGTS
with a paricular property
\emph{perfect}~\cite{DBLP:journals/tcs/Lambert92}. Since the precise
definition is involved and we do not need all the details, it is
enough for us to mention a selection of properties of perfect MGTS.
Intuitively, in perfect MGTSes, the value $\omega$ on place $p$ in
$m_i$ means that inside of the component $C_i$, the token count in
place $p$ can be made arbitrarily high.
In~\cite{DBLP:journals/tcs/Lambert92} it is shown (Theorem 4.2
(page~94) together with the preceding definition) that

\begin{thm}[\cite{DBLP:journals/tcs/Lambert92}]\label{thm:language-union}
For a Petri net $N$ one can compute finitely many perfect MGTS $\Ncal_1, \ldots, \Ncal_m$
such that $L(N,M_I,M_F) = \bigcup_{i=1}^m L(\Ncal_i)$.
\end{thm}

Moreover, by Corollary~4.1 in~\cite{DBLP:journals/tcs/Lambert92} (page~93), given a perfect
MGTS $\Ncal$, it is decidable whether $L(\Ncal)\ne\emptyset$. Therefore, our task
reduces to the following. We have a perfect MGTS $\Ncal$ with $L(\Ncal)\ne\emptyset$
and want to compute regular languages $R_1,\ldots,R_k$ such that
$L(\Ncal)\subseteq R_1\cdots R_k$ and $R_1\times\cdots\times R_k\subseteq \factors[k]{L(\Ncal)}$.
(Note that if the MGTS have different lengths, we can always fill up with $\{\eps\}$).
We choose $R_1,\ldots,R_k$ to be the sequence $L(C_0), \{t_1\},
L(C_1),\ldots,L(C_{n-1}),\{t_n\},L(C_n)$. Then \cref{eq:overapproximation} tells us
that this achieves $L(\Ncal)\subseteq R_1\cdots R_k$ and all that remains to be
shown is
\begin{equation}
L(C_0)\times\{t_1\}\times L(C_1) \times \cdots \times L(C_{n-1}) \times\{t_n\} \times L(C_n)\subseteq \factors[2n+1]{L(\Ncal)}.\label{eq:inclusion-mgts}
\end{equation}

\subparagraph*{Constructing runs}
In order to show \cref{eq:inclusion-mgts}, we employ a simplified version of
Lambert's iteration lemma, which involves covering sequences.
Let $C$ be a precovering graph for a Petri net $N = (P, T, \Pre, \Post)$
with a distinguished vector $m \in \Nom^P$ and initial vector $m^{\init} \in \Nom^P$.
A sequence $u \in L(C) \cap L(N, m^{\init})$ is called a \emph{covering sequence for $C$}
if for every place $p \in P$ we have either 1) $m^{\init}[p] = \omega$, or 2) $m[p] = m^{\init}[p]$ and $\eff(u)[p] = 0$,
or 3) $m[p] = \omega$ and $\eff(u)[p] > 0$.
This corresponds intuitively to the three possible cases for the set of runs in $N$ crossing the component $C$ in a place $p$:
(i)~runs that can have arbitrarily high value on $p$ when entering $C$,
(ii)~runs where, when entering $C$, $p$ has a fixed value, and the tokens in $p$ cannot be pumped inside of $C$,
or (iii)~runs where, when entering $C$, $p$ has a fixed value, but it can be pumped up inside of $C$.

Let $\Ncal = C_0, t_1, C_1 \ldots C_{n-1}, t_n, C_n$ be an MGTS, where $C_i = (V_i, E_i, h_i)$ is a precovering graph, and
let the distinguished vertex be $m_i$ and initial vertex be $m_i^\init$. If $\Ncal$ is a perfect MGTS
then according to the definition from~\cite{DBLP:journals/tcs/Lambert92} (page~92),
for every $i \in [0,n]$ there exists a covering sequence $u_i \in L(C_i) \cap L(N, m_i^\init)$. 
This corresponds to the mentioned intuition that $\omega$ values imply arbitrarily high values.
As a direct consequence of Lemma 4.1 in~\cite{DBLP:journals/tcs/Lambert92} (page~92), Lambert's iteration lemma, we obtain:
\begin{lem}\label{lem:iteration}
Let $\Ncal = C_0, t_1, C_1 \ldots C_{n-1}, t_n, C_n$ be a perfect MGTS and let
$x_i$ be a covering sequences for $C_i$ for $i\in[0,n]$. Then there exist
words $y_i \in T^*$ for $i \in [0,n]$ such that
$x_0 y_0 \cdot t_1 \cdot x_1 y_1 \cdots x_{n-1} y_{n-1} \cdot t_n \cdot x_n y_n \in L(\Ncal)$.
\end{lem}
%
  \cref{lem:iteration} is obtained from Lemma 4.1
  in~\cite{DBLP:journals/tcs/Lambert92} as follows. The word $u_i$ there is our
  $x_i$ and $v_i$ there is an arbitrary covering sequence of $C_i$ reversed. Then,
  our $y_i$ is set to $u_i^{k-1}\beta_i (w_i)^k (v_i)^k$ for some $k\ge k_0$.
%
The only technical part of the proof of~\cref{thm:approximation} is the following lemma.

\begin{lem}\label{lem:covering-sequences}
Let $C$ be a precovering graph for a Petri net $N = (P, T, \Pre, \Post)$
with a distinguished vector $m \in \Nom^P$ and initial vector $m^\init \in \Nom^P$
such that $s \in L(C) \cap L(N, m^\init)$ is a covering sequence.
Then for every $v \in L(C)$ there is a covering sequence for $C$ of the form $u v$,
for some $u \in T^*$.
\end{lem}

\begin{proof}
  Intuitively, we do the following. The existence of a covering
  sequence means that one can obtain arbitrarily high values on places $p$ where
  $m[p] = \omega$.  Thus, in order to construct a covering sequence containing $v$
  as a suffix, we first go very high on the $\omega$ places, so high that
  adding $v$ as a suffix later will still result in a sequence with positive effect.  

  Let us make this precise. Executing the sequence $v$ might have a negative
  effect in a place $p\in P$ with $m[p]=\omega$. Let $k\in\N$ be the largest
  possible negative effect a prefix of $v$ can have in any coordinate. Note that
  since $s$ is a covering sequence, $s^k$ is a covering sequence as well. We claim
  that $s^kv$ is also a covering sequence. It is contained in $L(C)$ and fireable
  at $m^{\init}$.  Moreover, by choice of $k$, the sequence $s^kv$ has a positive
  effect on each $p$ with $m[p]=\omega$. If $m[p]<\omega$, then
  $\eff(s)[p]=0=\eff(v)[p]$ and hence $\eff(s^kv)[p]=0$.
\end{proof}

Using \cref{lem:iteration} and \cref{lem:covering-sequences}, it is now easy to
show \cref{eq:inclusion-mgts}. Given words $v_i\in T^*$ with $v_i\in L(C_i)$
for $i\in[0,n]$, we use \cref{lem:covering-sequences} to choose $x_i\in T^*$ such
that $x_iv_i$ is a covering sequence of $C_i$ for $i\in[0,n]$. By \cref{lem:iteration},
we can find $w_1,\ldots,w_n$ so that
\[ x_0v_0w_0\cdot t_1\cdot x_1v_1w_1 \cdots x_{n-1}v_{n-1}w_{n-1} \cdot t_n\cdot x_nv_nw_n \in L(\Ncal), \]
and thus $(v_0,t_1,v_1,\ldots,v_{n-1},t_n,v_n)\in\factors[2n+1]{L(\Ncal)}$, which proves \cref{eq:inclusion-mgts}.

\subparagraph*{Acknowledgements}  We
are indebted to Mohamed Faouzi Atig for suggesting to study
separability by bounded languages, which was the starting point for
this work. Furthermore, we would like to thank S{\l}awomir Lasota and
Sylvain Schmitz for important discussions.  Finally, we are happy to
acknowledge that this collaboration started at the Gregynog~71717
research workshop organized by Ranko Lazi\'{c} and Patrick Totzke.

\bibliographystyle{plain}
\bibliography{bib}

\appendix


\section{Separability by bounded regular languages}\label{sec:separability}
This section contains the omitted proofs concerning separability by bounded regular languages.

\begin{proof}[Proof of \cref{separability-commutative}]
First, if $L_0$ and $L_1$ are separable by a regular $R\subseteq
w_1^*\cdots w_n^*$, then the set 
\[ S=\{(x_1,\ldots,x_n) \mid w_1^{x_1}\cdots w_n^{x_n}\in R\} \]
is recognizable. This is a classical result by Ginsburg and
Spanier~\cite{GinsburgSpanier1966a}. Moreover, $S$ clearly separates
$U_0$ from $U_1$.

Conversely, if $S\subseteq \N^n$ is recognizable and separates $U_0$
from $U_1$, then the set 
\[ R=\{w_1^{x_1}\cdots w_n^{x_n} \mid (x_1,\ldots,x_n)\in S \} \]
is
regular. Let us show that it separates $L_0$ and $L_1$.  If $w\in
L_0$, then we can write $w=w_1^{x_1}\cdots w_n^{x_n}$, which implies
$(x_1,\ldots,x_n)\in U_0$.  Therefore, we have $(x_1,\ldots,x_n)\in S$
and thus $w=w_1^{x_1}\cdots w_n^{x_n}\in R$. Thus, $L_0\subseteq
R$. Now suppose $w\in R$. Then we can write $w=w_1^{x_1}\cdots
w_n^{x_n}$ with $(x_1,\ldots,x_n)\in S$. That implies
$(x_1,\ldots,x_n)\notin U_1$ and hence $w_1^{x_1}\cdots
w_n^{x_n}\notin L_1$. Hence, $R\cap L_1=\emptyset$.
\end{proof}

In the proof, we also use the following fact:

\begin{prop}
If $L\subseteq w_1^*\cdots w_n^*$ is a VAS language, then
the set \[ U=\{(x_1,\ldots,x_n)\in\N^n \mid w_1^{x_1}\cdots w_n^{x_n}\in L\} \]
is a effectively a section of a VAS reachability set.
\end{prop}

\begin{proof}
First recall the notion of a section.
For a subset $I\subseteq [1,n]$, let $\pi_I\colon \N^n\to\N^{|I|}$ be the projection onto the coordinates in $I$.
Then, every set of the form $\pi_{[1,n]\setminus I}(S\cap \pi_{I}^{-1}(x))$ for some
$I\subseteq[1,n]$ and $x\in\N^{|I|}$ is called a \emph{section} of $S\subseteq\N^n$.
Intuitively, we fix a vector $x \in \N^{|I|}$ on coordinates from $I$ and take into the section all the
vectors $y \in \N^{n-|I|}$, which together with $x$ form an $n$-dimensional vector from $S$.

Assume that $L$ is a language of $d$-dimensional VAS $V$.
In order to show that $U$ is a section of a VAS reachability set we construct another VAS $V'$
in the following way. VAS $V'$ simulates $V$ on $d$ coordinates and has $n$ additional coordinates,
on which it counts number of occurrences of words $w_1, \ldots, w_n$. It is easy to see that VAS indeed can count
such occurrences by keeping some additional finite information, like the suffix of current run, which has not been yet
counted into any $w_i$ and the information which $w_i$ has recently appeared.
Section of $V'$ leaving only these $n$ counting coordinates is exactly the set $U$.
\end{proof}


\section{Downward closures}

\begin{proof}[Proof of \cref{lemma:sup}]
  Since $a_1^*\times\cdots \times
  a_n^*\subseteq\dcl{\factors[n]{L_1\cdots L_k}}$, we know that for
  every $\ell\in\N$, we can find words $w_i\in L_i$ for $i\in[1,k]$ so
  that $a_1^{\ell\cdot k}\cdots a_n^{\ell\cdot k}\subword w_1\cdots
  w_k$. Then, in particular, there is a monotone map
  $\pi_\ell\colon[1,n]\to[1,k]$ so that $a_i^\ell\subword w_{\pi(i)}$.
  Since there are only finitely many maps $[1,n]\to[1,k]$, there is
  one monotone map $\pi\colon[1,n]\to[1,k]$ that occurs infinitely
  often in the sequence $\pi_1,\pi_2,\ldots$. We can decompose
  $[1,n]=\pi^{-1}(1)\cup\cdots\cup\pi^{-1}(k)$ and since $\pi$ is
  monotone, each $\pi^{-1}(i)$ is convex. This give rise to a
  decomposition $n=n_1+\cdots+n_k$ so that $\pi^{-1}(1)\subseteq[1,n_1]$,
  $\pi^{-1}(2)\subseteq[n_1+1,n_2]$, etc.  Now, by choice of $\pi$,
  for each $\ell\in\N$, we can find  $w_i\in L_i$ so
  that $a_i^\ell\subword w_{\pi(i)}$, which means
  $(a_1^\ell,\ldots,a_n^\ell)\in
  \dcl{(\factors[n_1]{L_1}\times\cdots\times\factors[n_k]{L_k})}$.
  This implies
  $a_1^*\times\cdots\times
  a_n^*\subseteq\dcl{(\factors[n_1]{L_1}\times\cdots\times\factors[n_k]{L_k})}$.
\end{proof}



\section{Non-overlapping factors}
\begin{proof}[Proof of \cref{factors-unboundedness-regular}]
  Suppose $K\subseteq\Sigma^*$ and let $\cA$ be a finite automaton for
  $R\subseteq\Sigma^*$. Pick a symbol $c\notin\Sigma$. We obtain a
  finite automaton $\cB$ from $\cA$ as follows. In the first step, for
  each pair $p,q$ of states, we check whether there is a word in $K$
  that labels a path $p$ to $q$ in $\cA$: This is decidable because we
  can effectively intersect languages in $\cC$ with regular languages
  and emptiness is decidable for $\cC$. If such a word exists, we add
  an edge labeled $c$ from $p$ to $q$. In the second step, for each
  edge with a label $\ne c$, we replace the label by $\eps$. This
  completes the construction of $\cB$.

  Clearly, $f_K$ is unbounded on $R$ if and only if $\{c\}^*\subseteq
  L(\cB)$.  Moreover, if $f_K$ is bounded on $R$, then $L(\cB)$ is
  finite and we can compute the maximal length $\ell$ of a word in
  $L(\cB)$.  This $\ell$ is then an upper bound for $f_K$ on $L$.
\end{proof}

\section{Counting automata}

We begin with a formal definition of the step relation in counting automata.
For a label
$x\in\Sigma\cup\{\varepsilon\}$, and configurations $(q,u,\mu),
(q',u',\mu')$, we write $(q,u,\mu)\autstep{x}(q',u',\mu')$ 
 if there is
an edge $(q,x,o,q)\in E$ such that one of the following holds:
\begin{itemize}
\item  We have $o=\oppush{a}$ for some $a\in\Gamma$ and  $u'=ua$ and $\mu'=\mu$.
\item 
We have 
$o=\opcheck{K}{c}$ for some $K\subseteq\Gamma^*$ from $\cC$ and
  $c\in C$ and $u'=\varepsilon$ and either (a)~$u\in K$ and
$\mu'=\mu+1_c$
 or (b)~$u\notin K$ and $\mu'=\mu$. Here,  $1_c\in\N^C$ is the vector with
$1_c(c)=1$ and $1_c(c')=0$ for $c'\in C\setminus c$. 
\end{itemize}
Moreover, for $w\in\Sigma^*$, we write $(q,u,\mu)\autstep{w} (q',u',\mu')$ 
if
\[ (q,u,\mu)=(q_1,u_1,\mu_1)\autstep{x_0}\cdots \autstep{x_n}(q_n,u_n,\mu_n)=(q',u',\mu'), \]
for some configurations $(q_i,u_i,\mu_i)$ and $w = x_0 \cdots x_n$, where $x_0, \ldots, x_n \in\Sigma_\eps$.

In our proof of \cref{counting-automata:decidable}, we will use
\cref{thm:approximation-unboundedness} and hence decidability of a multidimensional
predicate. Suppose $t=(K_1,\ldots,K_n)$ is a tuple of languages
$K_i\subseteq\Sigma^+$.  We define a function $f_t\colon\Sigma^*\to\N$
as follows.  Intuitively, $f_t(w)$ is the largest number $k$ so that
we can pick a set of non-overlapping factors of $w$ among whom there
are at least $k$ members of $K_i$ for each $i\in[1,n]$.  

Formally, for a word $w\in\Sigma^*$, let $f_t(w)$ be the largest
number $\ell$ such that there is a tuple $(w_1,\ldots,w_m)\in
(\Sigma^+)^m$ with $(w_1,\ldots,w_m)\in \factors[m]{w}$ such that for
each $i\in[1,n]$, we have $|\{j\in[1,m] \mid w_j\in K_i\}|\ge \ell$.
Using an $n$-dimensional predicate and \cref{thm:approximation-unboundedness}, we
can show the following.
\begin{lem}\label{non-overlapping-factors:multi:decidable}
  Let $\cC$ be a full trio with decidable emptiness. Given a tuple
  $t=(K_1,\ldots,K_n)$ of languages from $\cC$ and a VAS language
  $L$, it is decidable whether $f_t$ is unbounded on $L$. Moreover, if
  $f_t$ is bounded on $L$, one can compute an upper bound $B\in\N$ for
  $f_t$ on $L$.
\end{lem}
\begin{proof}
  Let $t=(K_1,\ldots,K_n)$ be a tuple of languages with
  $K_i\subseteq\Sigma^+$ for $i\in[1,n]$.  For a word $w\in\Sigma^*$,
  let $\Delta(w)\subseteq\N^n$ be the set of all
  $(x_1,\ldots,x_n)\in\N^n$ such that there is a tuple $(w_1,\ldots,
  w_m)\in (\Sigma^+)^m$ with $(w_1,\ldots,w_m)\in \factors[m]{w}$ and
  $x_i=|\{j\in[1,m]\mid w_j\in K_i\}|$.

  Let us now define the predicate $\pP$.  For
  $S\subseteq(\Sigma^*)^n$, let $\pP(S)$ express that for every
  $\ell\in\N$, there is a tuple $(w_1,\ldots,w_n)\in S$ and a vector
  $(x_1,\ldots,x_n)\in \sum_{i=1}^n \Delta(w_i)$ such that
  $x_i\ge\ell$ for each $i\in[1,n]$. Here, the sum on subsets of
  $\N^n$ is to be read as the Minkowski sum: $X+Y=\{x+y \mid x\in
  X,~y\in Y\}$. Note that then indeed $\pP(\factors[n]{L})$ if and
  only if $f_t$ is unbounded on $L$.

  The predicate $\pP$ clearly satisfies
  \cref{axiom:upward,axiom:union}, so let us prove
  \cref{axiom:concatenation} and suppose $\pP(\factors[n]{L_1\cdots
    L_k})$. A \emph{profile} is a map $\pi\colon
  [1,n]\to[1,k]$. Intuitively, a profile records for each $i\in[1,n]$
  which of the factors $L_1,\ldots,L_k$ can be chosen to find a
  particular number of factors from $K_i$.

  Let $\ell\in\N$. Since $\pP(\factors[n]{L_1\cdots L_k})$, we know
  that there is a $(w_1,\ldots,w_n)\in\factors[n]{L_1\cdots L_k}$ such
  that there is a $(x_1,\ldots,x_n)\in\sum_{i=1}^n \Delta(w_i)$ with
  $x_i\ge k\cdot\ell+k$. Since $(w_1,\ldots,w_n)\in\factors[n]{L_1\cdots L_k}$, there is a word $u\in
  L_1\cdots L_k$ with $(w_1,\ldots,w_n)\in\factors[n]{u}$.
  Thus, we have a $(y_1,\ldots,y_n)\in\Delta(u)$ with $y_i\ge k\cdot \ell+k$
  for each $i\in[1,n]$. Since $u\in L_1\cdots L_k$, we can write
  $u=u_1\cdots u_k$ with $u_i\in L_i$ for $i\in[1,k]$.

  Observe that then there is a $(z_1,\ldots,z_n)\in\sum_{i=1}^k
  \Delta(u_i)$ with $z_i\ge y_i-k$ for $i\in[1,n]$: From the set of
  factors that witnesses $(y_1,\ldots,y_n)\in\Delta(u)$, we can select
  those that are confined to a single $u_i$; then we lose at most
  those that fall on the border of two $u_i$'s, hence at most
  $k$. Since $y_i\ge k\cdot\ell+k$, we have $z_i\ge k\cdot \ell$ for
  $i\in[1,n]$.  Write $(z_1,\ldots,z_n)=\sum_{i=1}^k
  (z_{i,1},\ldots,z_{i,n})$ with
  $(z_{i,1},\ldots,z_{i,n})\in\Delta(u_i)$. Since
  $z_{1,i}+\cdots+z_{k,i}=z_i\ge k\cdot\ell$, we can find for each
  $i\in[1,n]$, an index $j\in[1,k]$ so that $z_{j,i}\ge \ell$.  This
  defines a profile $\pi_\ell$: Let $\pi_\ell(i)=j$. 

  To summarize, we have defined for each $\ell\in\N$ a profile
  $\pi_\ell$ so that the following holds. For each $\ell\in\N$, there
  are words $u_1,\ldots,u_k$ with $u_j\in L_j$ for $j\in[1,k]$ so that
  for each $i\in[1,n]$, the set $\Delta(u_{\pi_\ell(i)})$ contains a
  vector $(z_1,\ldots,z_n)$ with $z_i\ge\ell$.

  Since there are only finitely many profiles, the sequence
  $\pi_1,\pi_2,\ldots$ must contain one profile $\pi$ infinitely
  often. This profile has thus the following property. For each $\ell\in\N$,
  there are words $u_1,\ldots,u_k$ with $u_j\in L_j$ for $j\in[1,k]$ so that
  for each $i\in[1,n]$, the set $\Delta(u_{\pi(i)})$ contains a
  vector $(z_1,\ldots,z_n)$ with $z_i\ge\ell$.
  
  This allows us to define the decomposition $n=n_1+\cdots+n_k$: For
  each $j\in[1,k]$, let $n_j=|\{i\in[1,n]\mid \pi(i)=j\}|$. We claim
  that then $\pP(\factors[n_1]{L_1}\times\cdots\factors[n_k]{L_k})$
  holds.  Let $j\in\N$. We can choose words $u_1,\ldots,u_k$ with
  $u_j\in L_j$ for $j\in[1,k]$ so that for each $i\in[1,n]$, the set
  $\Delta(u_{\pi(i)})$ contains a vector $(z_1,\ldots,z_n)$ with
  $z_i\ge\ell$. 

  Let us construct the tuple $(v_1,\ldots,v_n)$ successively from left
  to right.  For each $j=1,\ldots,k$, we do the following. If $n_j=0$,
  then we add no new component.  If $n_j>0$, then we include $u_j$ and
  then $(n_j-1)$ entries containing just the empty word
  $\varepsilon$. This clearly yields a tuple with $n=n_1+\cdots+n_k$
  entries. Moreover, we have
  $(v_1,\ldots,v_n)\in\factors[n_1]{L_1}\times\cdots\times\factors[n_k]{L_k}$.
  Finally, for each $i\in[1,n]$, we have $n_{\pi(i)}>0$ and hence
  $u_{\pi(i)}$ occurs in the tuple $(v_1,\ldots,v_n)$. Therefore, some
  $\Delta(v_i)$ contains a vector $(z_1,\ldots,z_n)$ with
  $z_i\ge\ell$.  Therefore, the sum $\sum_{i=1}^n\Delta(v_i)$ contains
  a tuple $(z_1,\ldots,z_n)$ with $z_i\ge \ell$ for every $i\in[1,n]$.
  This proves our claim and hence that $\pP$ satisfies
  \cref{axiom:concatenation}.
  This shows that $\pP$ is in fact an unboundedness predicate.

  According to \cref{thm:approximation-unboundedness}, we can compute
  a regular language $R\supseteq L$ such that $\pP(\factors[n]{L})$ if
  and only if $\pP(\factors[n]{R})$. This means $f_t$ is unbounded on
  $L$ if and only if it is unbounded on $R$. Moreover, since
  $L\subseteq R$, an upper bound of $f_t$ on $R$ is also an upper
  bound of $f_t$ on $L$. Thus, it remains to show that we can decide
  whether $f_t$ is bounded on $R$ and, if so, we can compute an upper
  bound of $f_t$ on $R$.

  Take a finite automaton $\cA$ for $R$.  From $\cA$, we obtain a
  finite automaton $\cB$ over the alphabet $\Gamma=\{a_1,\ldots,a_n\}$
  as follows. First, we remove all edges.  Then, for each pair $p,q$
  of states and each $i\in[1,n]$, we check whether there is a word
  $K_i$ that is read on a path from $p$ to $q$ in $\cA$: This can be
  checked because $K_i$ belongs to $\cC$, $\cC$ is effectively closed
  under intersecion with regular languages, and emptiness is decidable
  for $\cC$. If that is the case, then we draw a new edge labeled
  $a_i$ from $p$ to $q$. Then, clearly, $f_t$ is unbounded on $R$ if
  and only if for every $\ell\in\N$, there is a word $w$ accepted by
  $\cB$ that contains $a_i$ at least $\ell$ times, for each
  $i\in[1,n]$.  Consider the set
  \[ S=\{\ell\in\N \mid \exists w\in L(\cB)\colon \forall i\in[1,n]\colon |w|_{a_i}\ge \ell \}. \]
  It is easy to see that $S$ is effectively semilinear: the Parikh
  image of $L(\cB)$ is semilinear and hence $S$ is definable in Presburger
  arithmetic. Furthermore, $f_t$ is unbounded on $R$ if and only if
  $S$ is infinite, which is easy to check. Finally, if $f_t$ is
  bounded on $R$, then $S$ is finite and we can compute the maximal
  element of $S$, which is an upper bound of $f_t$ on $R$.
\end{proof}

In the proof of \cref{counting-automata:decidable}, we will use the
concept of a transducer. A \emph{(finite-state) transducer} is a tuple
$\cA=(Q,\Sigma,\Gamma,E,q_0,Q_f)$, where $Q$ is a finite set of
\emph{states}, $\Sigma$ is its \emph{input alphabet}, $\Gamma$ is its
\emph{output alphabet}, $E\subseteq Q\times
\Sigma_\eps\times\Gamma_\eps\times Q$ is its set of \emph{edges},
$q_0\in Q$ is its \emph{initial state}, and $Q_f\subseteq Q$ is its
set of \emph{final states}. A \emph{configuration} of $\cA$ is a
triple $(q,u,v)\in Q\times\Sigma^*\times\Gamma^*$ and we write
$(q,u,v)\to(q',u',v')$ if there is an edge $(q,x,y,q')$ with $u'=ux$
and $v'=vy$. Let $\to^*$ denote the reflexive transitive closure of
$\to$. 

Subsets of $\Sigma^*\times\Gamma^*$ for alphabets $\Sigma,\Gamma$ are
called \emph{transductions}.
A transducer induces a transduction as follows:
\[ T(\cA)=\{(u,v)\in\Sigma^*\times\Gamma^* \mid (q_0,\eps,\eps)\to^* (q,u,v)~\text{for some $q\in Q_f$} \}. \]
Then, $T(\cA)$ is called the transduction \emph{induced by $\cA$}.  
A transduction of the form $T(\cA)$ is called a \emph{rational transduction}.
In
general, for a transduction $T\subseteq\Sigma^*\times\Gamma^*$ and a
language $L\subseteq\Sigma^*$, we define
\[ T(L)=\{v\in\Gamma^* \mid \exists u\in L\colon (u,v)\in T \}. \]
It is well known that a language class $\cC$ is a full trio if and
only if it is effectively closed under rational transductions, meaning
given a description of $L$, we can effectively compute a description
of $T(L)$ in $\cC$.

We are now ready to prove \cref{counting-automata:decidable}.
\begin{proof}[Proof of \cref{counting-automata:decidable}]
  Given $\cA$, we can transform $\cA$ into a transducer $\cB$ as
  follows.  Let $K_1,\ldots,K_n$ be the languages occurring in edges
  $\opcheck{K}{c}$ in $\cA$ and pick letters $d,e_{i,c}\notin\Gamma$
  for each $i\in[1,n]$ and $c\in C$. The transducer $\cB$ operates like $\cA$, but
  instead of performing operations $\oppush{a}$ or $\opcheck{K_i}{c}$,
  it outputs symbols from the alphabet $\Lambda=\Gamma\cup
  \{d,e_{i,c}\mid i\in[1,n], c\in C\}$: When $\cA$ performs
  $\oppush{a}$, $\cB$ outputs $a$.  When $\cA$ performs
  $\opcheck{K_i}{c}$, then $\cB$ outputs $e_{i,c}d$. Moreover, in the
  beginning of a run, $\cB$ outputs a single $d$ before it starts
  operating like $\cA$. Now let $T$ be the transduction induced by
  $\cB$ and let $L'=T(L)$. 
  Then $L'$ is again a VAS language and
  consists of precisely those words
  $du_1e_{i_1,c_1}du_2e_{i_2,c_2}\cdots du_me_{i_m,c_m}u$ such
  that $u\in\Gamma^*$ and $\cA$ has a run on a member of $L$ that
  performs for each $j\in[1,m]$ the operation
  $\opcheck{K_{i_j}}{c_j}$ while $u_j$ is on the work tape.

  Consider the language class $\bar{\cC}$, which consists of all
  finite unions of languages in $\cC$. Then $\bar{\cC}$ is again a
  full trio and has a decidable emptiness problem.  For each $c\in C$,
  let $\bar{K}_c=\bigcup_{i\in[1,n]} dK_ie_{i,c}$.  Then clearly
  $\bar{K}_c$ belongs to $\bar{\cC}$. Let $C=\{c_1,\ldots,c_k\}$ and
  consider the language tuple
  $t=(\bar{K}_{c_1},\ldots,\bar{K}_{c_k})$. Then $f_t$ is unbounded on
  $L'$ if and only if $\cA$ is unbounded on $L$. Moreover, an upper
  bound $B\in\N$ for $f_t$ on $L'$ is also an upper bound for $\cA$ on
  $L$.  Thus, an application of
  \cref{non-overlapping-factors:multi:decidable} completes the proof.
\end{proof}


\section{Factor inclusion}

\subparagraph*{Detailed proof of \cref{axiom:union}}
First, let us verify \cref{axiom:union} in detail.  Suppose that
$L_1\cup L_2$ is $K$-factor universal and that $L_1$ is not $K$-factor
universal. The latter means there is some $u\in K^*$ with
$u\notin\factors{L_1}$. Now let $v\in K^*$ be arbitrary. Since $uv\in
K^*$ and by $K$-factor universality of $L_1\cup L_2$, we know that
$uv\in\factors{L_1\cup L_2}=\factors{L_1}\cup\factors{L_2}$. Since
$uv\in\factors{L_1}$ is impossible, this only leaves
$uv\in\factors{L_2}$ and in particular $v\in\factors{L_2}$. This
proves that $L_2$ is $K$-factor universal and hence
\cref{axiom:union}.

It remains to show decidability of whether  $K^*\subseteq\factors{R}$.

\begin{lem}\label{factor-universality-regular}
  Let $\cC$ be a full trio with decidable emptiness. Given a language
  $K$ from $\cC$ and a regular language $R$, it is decidable whether
  $K^*\subseteq\factors{R}$.
\end{lem}
\begin{proof}
  Suppose $K\subseteq\Sigma^*$ and let $\cA$ be a finite automaton for
  the regular language $\Sigma^*\setminus\factors{R}$. We have to
  decide whether $K^*\cap L(\cA)=\emptyset$.  Pick a symbol
  $c\notin\Sigma$. We obtain a finite automaton $\cB$ from $\cA$ as
  follows. For each pair $p,q$ of states, we check whether there is a
  word in $K$ that labels a path $p$ to $q$ in $\cA$: This is
  decidable because we can effectively intersect languages in $\cC$
  with regular languages and emptiness is decidable for $\cC$. If such
  a word exists, we add an edge labeled $c$ from $p$ to $q$. In the
  second step, we remove all edges except for those labeled $c$. This
  finishes the construction of $\cB$.  Then we have
  $L(\cB)\subseteq\{c\}^*$.  Furthermore, $K^*$ intersects $L(\cA)$ if
  and only if $L(\cB)\ne\emptyset$.
\end{proof}


\end{document}